%% file: RDC.tex
\documentclass[11pt]{article}
\usepackage{microtype}
\usepackage{amsmath,amsfonts,amsthm,amssymb}
\usepackage{fancyhdr}
\usepackage{extramarks}
\usepackage{chngpage}
\usepackage{soul,color}
\usepackage{graphicx,float,wrapfig}
\usepackage{fullpage}
\usepackage{algorithm}
\usepackage{algorithmic}
\newcommand{\eat}[1]{}
\newcommand{\PRP}{\textsf{PRP}}
\newcommand{\RMI}{\textsf{R-Multi}}
\newcommand{\RUI}{\textsf{R-Unrel}}
\newcommand{\RSI}{\textsf{R-Submod}}
\newcommand{\RDC}{\textsf{RDC}}
\newcommand{\RGC}{\textsf{RGC}}
\newcommand{\RXOS}{\textsf{RXOS}}
\newcommand{\LP}{\mathsf{LP}}
\newcommand{\SC}{\mathsf{SC}}
\newcommand{\OPT}{\mathsf{OPT}}

\newcommand{\LPOpt}{\mathsf{OPT}_{\mathsf{LP}}}
\newtheorem{theorem}{Theorem}
\newtheorem{lemma}[theorem]{Lemma}
\newtheorem{definition}[theorem]{Definition}

\newtheorem{example}[theorem]{Example}

\title{Ranking with Diverse Intents and Correlated Contents}
\author{Jian Li\\ lijian83@mail.tsinghua.edu.cn\\ IIIS\\ Tsinghua University
        \and Zeyu Zhang\\ zyzhang92@gmail.com\\ Department of Mathematics\\ Tsinghua University}

\begin{document}
\maketitle

\begin{abstract}
We consider the following document ranking problem: We have a collection of documents, each containing some topics (e.g. sports, politics, economics). We also have a set of users with diverse interests. Assume that user $u$ is interested in a subset $I_u$ of topics. Each user $u$ is also associated with a positive integer $K_u$,
which indicates that $u$ can be satisfied
by any $K_u$ topics in $I_u$. Each document $s$ contains information for a subset $C_s$ of topics. The objective is to pick one document at a time such that the average \textsl{satisfying time} is minimized, where a user's satisfying time is the first time that at least $K_u$ topics in $I_u$ are covered in the documents selected so far.

Our main result is an $O(\rho)$-approximation algorithm for the problem, where $\rho$ is the algorithmic integrality gap of the linear programming relaxation of the set cover instance defined by the documents and topics. This result generalizes the constant approximations for generalized min-sum set cover and ranking with unrelated intents and the logarithmic approximation for the problem of ranking with submodular valuations (when the submodular function is the coverage function),
and can be seen as an interpolation between these results.
We further extend our model to the case when each user may interest in more than one sets of topics and when the user's valuation function is XOS, and obtain similar results for these models.
\end{abstract}

\input{intro}

\input{related}
\input{algo}
\input{extensions}

\section{Final Remarks}
As we mentioned in the introduction, 
the real world document-topic instance do not form arbitrary set system
and may be easier to approximate than the general combinatorial set cover problem.
There is a huge literature on algorithms for classifying or clustering the documents and modeling document-topic relations.
Many of those works leverage the underlining special structure of the document-topic instance to achieve good classification or clustering. It is an interesting further direction to explore the connections to those works and see whether 
the assumptions made or the special structures used in those works would translate to interesting set cover instances
that are easier to approximate.

We could extend our model in several ways to capture other factors
that may affect the search result.
For example, we can capture that each user only has limited patience in the following variant.
For each user $i$, we have a patience level $t_{i}$ that the maximum number of documents
user $i$ will examine. If the user does not get a relevant document after examining $t_{i}$ documents, she will leave the system. Now the objective is to maximize the expected number of satisfied users. This generalizes the traditional scheduling problem with deadlines. We can also incorporate uncertainty into the user preferences.
Namely, a user is interested in a particular document with a certain probability.
The resulting stochastic version of the problem
has a similar flavor with the sequential trial optimization defined in \cite{cohen2003efficient} or
the stochastic matching problem in \cite{chen2009approximating,bansal2012lp}.

Finally, we note that our approximation algorithm is mainly of theoretical interests since we need to use
the ellipsoid algorithm to solve a linear program with exponential constraints, which is computationally expensive
in practice. Hence, developing more efficient
algorithms for \RDC\ (even with worse performance guarantee) is of great practical interests.
\vspace{-0.1cm}
\subparagraph*{Acknowledgements}

Jian Li would like thank Yossi Azar for a stimulating discussion about the \RDC\ model
and for providing the manuscript \cite{azar2010rankingunrelated}.
\vspace{-0.1cm}
\bibliographystyle{plain}
\bibliography{ranking}
\end{document}

%% file: intro.tex
\section{Introduction}

\subsection{Background}
In a typical information retrieval application, we have a set of users and a set of documents.
Each user issues a query and we would like to present the user with a rank list of the documents.
Hopefully, the top-ranked documents are relevant to the user and our general objective
is to maximize the overall user satisfaction.
In many IR applications, the {\em probabilistic ranking principle} (\PRP)
is considered as a common rule of thumb to rank the documents
~\cite{robertson1993probability}.
\PRP\ states that we should rank the documents in descending order by their probability of relevance and
it is the ``optimal'' way to rank the documents in the sense that
\PRP\ minimizes the expected loss (also known as the  Bayes risk) under 1/0 loss~\cite{manning2008introduction}.
However, the 0/1 loss metric does not directly relate to the users' satisfaction
and sometimes the ranking given by \PRP\ is clearly suboptimal.
Indeed, even the original paper \cite{robertson1993probability} provided
such an example (the example was discovered by W.S.Cooper).

\begin{example}
\label{ex:1}
\cite{robertson1993probability}
The class of users consists of two subclasses $\mathcal{U}_1$ and $\mathcal{U}_2$.
$\mathcal{U}_1$ has 100 users and $\mathcal{U}_2$ has 50 users.
Any user from $\mathcal{U}_1$ would be satisfied with any document $s_1$--$s_9$,
but no others.
Any user from $\mathcal{U}_2$ would be satisfied with only $s_{10}$.
If we consider any document $s_1$--$s_9$ on its own,
it has a probability of $2/3$ of being relevant to the next user
(the ranking algorithm does not know which subclass the user belongs to).
Similarly, $s_{10}$ has a probability of $1/3$ of being relevant.
Therefore, by \PRP, the ranking should be $s_1, s_2, \ldots, s_9, s_{10}$.
But this means that $\mathcal{U}_1$ users can be satisfied with $s_1$ while $\mathcal{U}_2$ users
have to see nine irrelevant documents before they retrieve $s_{10}$.
Consider the ranking $s_1, s_{10}, s_2, \ldots, s_9$.
$\mathcal{U}_1$ users are still satisfied by the first document, but $\mathcal{U}_2$
users are satisfied with the second document, which is much better
than the ranking defined by \PRP.
\end{example}

The action of placing several documents aiming at different types of users
at the top positions of the rank list (e.g. place $s_{1}$ and $s_{2}$ as the top-2 in the above example)
is called {\em diversification}.
It is a widely accepted fact that {\em diversification} of the ranking result is helpful
in minimizing the risk of user dissatisfaction
in a multiuser scenario
(See, e.g., \cite{carbonell1998use,chen2007addressing,radlinski2008learning, agrawal2009diversifying, sigir09wang, drosou2010search, gollapudi2009axiomatic,azar2009multiple}).
Example~\ref{ex:1} is a simple yet instructive illustration why the diverse intents
and the correlations of the documents
($s_1$--$s_9$ are correlated in a way that any of them could satisfy a $\mathcal{U}_1$ user)
are the major reasons for diversification.

\begin{enumerate}
\item \underline{Diverse intentions:}
Different users may have different intents towards the same query (e.g., a keyword).
However, the ranking algorithm does not know the actual type of an individual user
but has to use the same ranking function for the same query.
In Example~\ref{ex:1}, there are two user types $\mathcal{U}_1$ and $\mathcal{U}_2$, and the next user could be either of them.
Considering another real life example, the keyword ``Michael Jordan" may refer to
the famous NBA player in one query, and the U.C. Berkeley Professor in another search.
\item \underline{Correlations among documents.}
Typically, the utility a user can obtain from a set of documents is not the sum of the utilities from individual documents in the set.
This is because of the similarity (or dissimilarity) of the documents. For instance, the utility of two very similar documents is not much more than
the utility of one of them
(e.g., documents $s_1$ and $s_2$ in Example~\ref{ex:1}).
Such correlations can be seen as another cause of diversification of the ranking result (see e.g.,~\cite{agrawal2009diversifying, sigir09wang, drosou2010search}).
\end{enumerate}

\subsection{Problem Formulation}

In this section, we propose our model for diversification, which captures both the diversity of users' intents
and the correlations of the documents.

\begin{definition}
\textsl{Ranking with Diverse Intents and Correlated Contents} (\RDC):
Here, we have a set $\mathcal{U}$ of users, a set $\mathcal{S}$ of documents, and a set $\mathcal{E}$ of topics. Each user $u$ is interested in a subset $I_u$ of topics. Each user $u$ is also associated with a positive integer $K_u$ which is less or equal to $|I_u|$. Each document $s$ contains a subset $C_s$ of topics and $\mathcal{E}=\bigcup_{s\in \mathcal{S}}C_s$. The objective is to pick an ordering of all documents
such that the average \textsl{satisfying time} is minimized, where a user's satisfying time
$t_u=\min\{t \mid \text{ at least }K_u\text{ topics in }I_u\text{ are covered by the first }t\text{ selected documents}\}.$
\end{definition}




It is not hard to see that our \RDC\ model captures both the diversity of the users' intents
(i.e., each user is interested in a different subset of topics)
and the correlations among documents
(i.e., different documents may have some common topics).
Now, we discuss some closely related prior work and their relations with our model.
\begin{enumerate}
\item \underline{Ranking with multiple intents (\RMI)} \cite{azar2009multiple}:
Azar et al. proposed the following combinatorial model
to capture the diversity of user preferences.
We have a set $\mathcal{U}$ of users
and a set $\mathcal{S}$ of documents.
User $u$ can be satisfied with any
$K_u$ document from a subset $I_u$ of documents.
The objective is the same as ours, to minimize the cumulative users' satisfying time.
We can see that it is a special case of \RDC\ where each document contains a distinct topic.
\item \underline{Ranking with unrelated intents (\RUI)} \cite{azar2010rankingunrelated}:
This model is a generalization of \RMI.
For each user $u$ and a document $s$,
there is a nonnegative number $A_{us}$ that is the amount of utility
$u$ can get from $s$.
$u$ is satisfied if she accumulates at least $K_u$ units of utility.
The objective is same as \RMI.
It is also not hard to see that
\RUI\ is a special case of \RDC\ where each document contains $A_{us}$ distinct topics.
\item \underline{Ranking with submodular intents (\RSI)} \cite{azar2011ranking}:
The model is a generalization of both \RMI\ and \RUI.
For each user $u$,
there is a nonnegative submodular function $f_u: \{0,1\}^{\mathcal{S}}\rightarrow \mathbb{R}^+\cup \{0\}$.
$u$ is satisfied if the set $S$ of documents she gets is such that $f_u(S)\geq 1$.
The objective is same as before.
\RSI\ generalizes \RDC\ (as well as \RMI\ and \RUI).
If the submodular function $f_u$ is the {\em coverage function}
\footnote{
The set of documents and the set of topics in $I_u$ form a set cover instance, where $I_u$ is the subset of topics which user $u$ is interested in.
For $S\subseteq \mathcal{S}$, the coverage function $f(S)$ is the number of topics in $I_u$ covered by some document in $S$.
}, \RSI\ is equivalent to \RDC.
\end{enumerate}

\subsection{Our results}

We find that the approximability of \RDC\ is closely related with the (algorithmic) integrality gap of the underlining set cover instance induced by the documents and topics. In particular, we can show the following result.
Let $F\subseteq \mathcal{E}$ be a subset of topics.  We denote $\SC(F)$ the set cover instance formed by the subsets $C_s:s\in \mathcal{S}$ and the set of topics in $F$.
Let $\LP(\SC(F))$ be the natural linear programming relaxation for $\SC(F)$:
\begin{eqnarray*}
\text{minimize}&\sum\limits_{s\in \mathcal{S}}x_s&\\
\text{subject to}&\sum\limits_{s:e\in C_s,s\in \mathcal{S}}x_s\ge1&\forall e\in F\\
&x_s\ge 0&\forall s\in \mathcal{S}
\end{eqnarray*}
\begin{theorem}
\label{thm:main}
Suppose for any $F\subseteq \mathcal{E}$, there is a polynomial time algorithm that can produce a solution for $\SC(F)$ whose cost is at most $\rho$ times the optimal value of
$\LP(\SC(F))$. There is a polynomial time factor $O(\rho)$ approximation algorithm for \RDC.
\end{theorem}

First, we can see that Theorem~\ref{thm:main} produces $O(1)$ factor approximation for both \RMI\ and \RUI.
As we mentioned before, if we view \RMI\ and \RUI\ as special cases of \RDC,
the induced set cover instances have very simple structure:
each set (document) consists of a disjoint set of elements (topics).
In both $\RMI$ and $\RDC$, the integrality gap of $\LP(\SC(F))$ is $1$ for any $F\subseteq \mathcal{E}$
and we can find an integral optimal solution in polynomial time
(the algorithm trivially includes all subsets that contains at least one element in $F$).
Hence, $\rho=1$ and we have a constant factor approximation algorithm.
Therefore, our result generalizes the constant approximations for \RMI\ in \cite{bansal2010constant,skutella2011note,im2012preemptive}
and that for \RUI\ in \cite{azar2010rankingunrelated}.

For \RSI, Azar et al. showed that there is an $O(\log \frac{1}{\epsilon})$ approximation for the problem
where $\epsilon$ is the minimum non-zero marginal value for $f_i$s \cite{azar2011ranking}.
If the submodular functions are coverage function, the result translates to an $O(\log |\mathcal{E}|)$-approximation.
It is well known that we can round any fractional solution of $\LP(\SC(F))$ to an integral solution
such that the cost of the integer solution is at most $\log |\mathcal{E}|$ of the value of the fractional solution.
Hence, Theorem~\ref{thm:main} also gives an $O(\log |\mathcal{E}|)$-approximation, reproducing the result in \cite{azar2011ranking} for \RSI\ with coverage functions
(with a somewhat larger constant hidden in the big-O notation).

Our result can be seen as an interpolation between the constant approximation for \RMI\ and \RUI\ (which induce trivial set cover instances)
and the logarithmic approximation for \RSI (which may induce arbitrary set cover instances).
Besides the above implications on previous problems, Theorem~\ref{thm:main}
is also interesting since typically the set cover instances induced by the documents and topics are much easier to approximate than general set cover problem.
We provide some useful examples.
\begin{enumerate}
\item In \RMI\ and \RUI, the induced set cover instances are trivial and can be solved optimally.
\item Consider another interesting example where each topic is covered by at most $d$ documents.
It is known that we can obtain a $d$-approximation by a simple deterministic rounding or primal-dual techniques (see e.g. \cite{vazirani2001approximation}).
Hence, in this case, we have an $O(d)$-approximation for \RDC.
\item
Suppose the VC dimension of the set system $(\mathcal{E}, \mathcal{S})$ is $d$. 
It is well known that we can achieve an approximation factor of $O(d\log \tau)$ via the LP approach~\cite{even2005hitting},
where $\tau$ is the optimum LP value
($O(d\log \OPT)$ is known even earlier via non-LP approach~\cite{bronnimann1995almost}).
In many cases, $O(d\log \tau)$ can be much smaller than $O(\log |\mathcal{E}|)$.
\item
For some geometric set cover problems, we can achieve sub-logarithmic factor approximation algorithms
using LP approaches. For example, if each subset corresponds to a unit disk in the plane and each element 
corresponds to a point, there is a constant approximation~\cite{pandit2009approximation}.
For general disk graphs, a $2^{O(\log^*|\mathcal{E}|)}$-approximation is known~\cite{varadarajan2010weighted}.

Several sub-logarithmic factor approximation algorithms are known for certain geometric set cover problems
via other techniques, such as local search or dynamic programming~\cite{ambuhl2006constant,mustafa2009ptas,gibson2010algorithms,chan2012approximation}.
However, it is not clear how to combine those techniques with our LP approach. We leave this as an interesting
open question.
\end{enumerate}
Even though the real world topics and documents may not necessarily have low VC-dimension or match
any geometric set cover instance,
it is still our general belief that the real world instance do not form arbitrary set system and
the particularity of those instances may help us to develop sub-logarithm factor approximations,
which further implies that \RDC\ can be approximated within the same factor (up to a constant).
Exploring the particularity of the real world instances is left as an open question of great importance.

%% file: related.tex
\subsection{Related work}

Azar et al. \cite{azar2009multiple} introduced \RMI\ and first gave an $O(\log n)$ factor approximation algorithm. 
Bansal et al. \cite{bansal2010constant} improved the approximation ratio to a constant (a few hundreds). 
Subsequently, the constant was further reduced to about 28 in~\cite{skutella2011note}, 
and then to 12.4~\cite{im2012preemptive}. 
An important special case of \RMI, where $K_u=1$ for each $u$, is called the {\em min-sum set cover} problem.
Feige et al.~\cite{feige2004approximating} developed a 4-approximation and proved that 
it is NP-hard to achieve an approximation factor of $4-\epsilon$ for any constant $\epsilon>0$.
In fact, it is conjectured that \RMI\ can also be approximated within a factor of $4$ \cite{im2012preemptive}.
Another special case of \RMI\ where $K_u=|I_u|$ has also been studied under the name of
{\em minimum latency set cover} and it is known that there is a polynomial time approximation algorithm with factor $2$
\cite{hassin2005approximation, karger1997scheduling}, which is also optimal assuming a variant of the Unique Games Conjecture
\cite{bansal2009optimal}.
Im et al.~\cite{im2012minimum} considered a generalization of \RSI\ where there is metric switching cost and gave
a poly-logarithmic factor approximation algorithm for it.

There is a huge literature on search result diversification in IR and DB literature.
We refer interested readers to 
\cite{carbonell1998use,chen2007addressing,agrawal2009diversifying, drosou2010search, gollapudi2009axiomatic,
borodin2012max} and the references therein.
In practice, the overall satisfying time as defined above
is not a direct measure of the overall user satisfaction.
Alternative measures have been proposed in the literature, such as discounted
cumulative gain (DCG) and mean average precision (MAP).
Bansal et al. considered \RMI\ with DCG being the objective function and obtained an $O(\log\log n)$-approximation~\cite{bansal2010approximation}.


%% file: algo.tex
\section{A Constant Factor Approximation Algorithm}

In this section, we will prove Theorem \ref{thm:main} by giving a randomized LP rounding algorithm.

\subsection{The LP Relaxation}

We use the following linear program relaxation.
Here we use boolean variable $x_{st}$ to represent whether document $s$ is selected at time $t$.
$y_{ut}$ indicates if user $u$ is satisfied after time $t$.
$z_{st}$ represents if document $s$ has been selected at time $t$.
\vspace{-0.2cm}
\begin{align}
\text{(\textbf{LP})}:\text{\ \ \ \ \ }\nonumber\\
\text{minimize \ }&\sum\limits_{u\in \mathcal{U}}\sum\limits_{t=1}^n(1-y_{ut})&&&\label{eq:0}\\
\text{subject to }&\sum\limits_{t=1}^nx_{st}=1&&\forall s\in \mathcal{S}&\label{eq:1}\\
&\sum\limits_{s\in \mathcal{S}}x_{st}=1&&\forall t\in [n]&\label{eq:2}\\
&z_{st}=\sum\limits_{t'=1}^tx_{st'}&&\forall t\in [n]&\label{eq:3}\\
&\sum\limits_{e\in I_u}(y_{ut}-\min\{\sum\limits_{s:e\in C_s}z_{st},1\})\le (|I_u|-K_u)y_{ut}&&\forall u\in \mathcal{U}, t\in [n]&\label{eq:4}\\
&x_{st},y_{ut},z_{st}\in [0,1]&&\forall s\in \mathcal{S}, u\in \mathcal{U}, t\in [n]&\label{eq:5}
\end{align}

Constraints \eqref{eq:1} and \eqref{eq:2} make sure that a document can be selected only once and each time we pick one document.
The meaning of $z_{st}$ is captured in constraints \eqref{eq:3}. Constraints \eqref{eq:4} guarantee that a user $u$ is satisfied if less than $|I_u|-K_u$ topics havn't been covered.
However, it is known that the integrality gap of this LP is unbounded (even for \RMI) \cite{bansal2010constant}.
To remedy this, \cite{bansal2010constant} uses the \textsl{knapsack cover} constraints to replace the simple covering constraints \eqref{eq:4}
In our case, we define $S(e,u,F)=\{s\mid e\in C_s,s\in T_2(u,F)\}$ where $T_1(u,F)$ is the set of all documents that cover at least $K_u-|F|$ topics in $I_u\backslash F$, and $T_2(u,F)=\mathcal{S}\backslash T_1(u,F)$.
And we use the following constraints instead of \eqref{eq:4}:
\vspace{-0.2cm}
\begin{align}
y_{ut}(K_u-|F|)\le (K_u-|F|)\sum\limits_{s\in T_1(u,F)}z_{st}+\sum\limits_{e\in I_u\backslash F}&\sum\limits_{s\in S(e,u,F)}z_{st}\nonumber\\
&\forall u\in \mathcal{U}, t\in [n], F\subseteq I_u,|F|\le K_u\label{eq:6}
\end{align}

Constraints \eqref{eq:6} differ from the knapsack cover constraints in \cite{bansal2010constant} in that we handle sets $T_1$ and $T_2$ seperately, for technical reason that will be clear from the analysis.


Now we show that (\textbf{LP}) is indeed an LP relaxation of \RDC. We just need to prove that any feasible solution to \RDC\ satisfies constraints \eqref{eq:6}: If $y_{ut}=0$, the inequality must be true because the left side is 0. If $y_{ut}=1$, there are two cases. The first case is that at least one document in $T_1(u,F)$ has been selected, which means $\sum\nolimits_{s\in T_1(u,F)}z_{st}\ge 1$. The other case is that at least $(K_u-|F|)$ topics have related document in $T_2(u,F)$, which means $\sum\nolimits_{e\in I_u\backslash F}\sum\nolimits_{s\in S(e,u,F)}z_{st}\ge(K_u-|F|)$. Therefore both cases satisfy the inequality.
So we have proved the following lemma:
\begin{lemma}
\label{lem:OPTeqRDC}
The optimal value $\LPOpt$ of (\textbf{LP}) is at most the optimal total satisfying time of \RDC.
\end{lemma}

\subsection{A Randomized Rounding Algorithm}

Assume $(x^*,y^*,z^*)$ be the optimal fractional solution to (\textbf{LP}).
We also assume that for any $F\subseteq \mathcal{E}$, there is a poly-time algorithm \textbf{AlgoSC} which
can produce an integral solution for $\SC(F)$ whose cost is at most $\rho$
times the value of the {\em fractional} optimal solution to $\LP(\SC(F))$.
Our randomized rounding scheme consists of
$\lceil\log n\rceil+1$ rounds, where in the $k$-th round, we perform the following procedure.

\begin{itemize}
\item Let $t=2^k$, $G_k=\emptyset$ and $p_e=\min\{1,50\sum\nolimits_{s:e\in C_s}z_{st}^*\}$, $\forall e\in \mathcal{E}$.
\item Let $P_k=\{e\in \mathcal{E},p_e=1\}$. Let the set $H_k\subseteq \mathcal{S}$ be the solution of \textbf{AlgoSC}($\SC(P_k)$).
\item For each $s\in \mathcal{S}\backslash H_k$, add document $s$ to $G_k$ independently with probability $\min\{1,50z_{st}^*\}$.
\item If there are more than $(70+\rho)\cdot 2^k$ documents in $H_k\cup G_k$, we say this round is "overflowed" and select nothing, else we select all the documents in $H_k\cup G_k$ in arbitrary order in this round.
\end{itemize}
Our algorithm builds on the ideas developed in \cite{bansal2010constant} (as well as \cite{azar2010rankingunrelated}).
A key technical difference between our algorithm and \cite{bansal2010constant}
is that we need to deal with those topics that are almost covered (i.e., the set $P_k$)
and the rest separately.
It will be clear soon from the analysis, for a particular user $u$, independent rounding (step 3) can guarantee that,
at a cost not much more than the fractional optimal,
topics in $I_u\backslash P_k$ are covered with constant probability.
For these topics,
we can use a Chernoff-like concentration result for submodular functions to show this.
Topics in $I_u\cap P_k$ are handled separately by \textbf{AlgoSC} to make sure they are covered in
the $k$-th round. This is where the approximiblity of the set cover instance jumps in.

\subsection{The Analysis}\label{subsec:analysis}

Constraints \eqref{eq:3} and \eqref{eq:5} show that the optimal solution $z_{st}^*$ is monotonically non-decreasing with $t$ for all $s\in \mathcal{S}$. Thus it's easy to see that $y_{ut}^*$ is monotonically non-decreasing with $t$ for all $u\in \mathcal{U}$.

For each $u\in \mathcal{U}$, let $t_u^*=\max\{t\in [n]\mid y_{ut}^*\le \frac{1}{2}\}$, then $\sum\nolimits_{t=1}^n(1-y_{ut}^*)\le\sum\nolimits_{t=1}^{t_u^*}(1-y_{ut}^*)\le\frac{1}{2}t_u^*$. Thus we have the fact that
$\LPOpt \ge \frac{1}{2}\sum\nolimits_ut_u^*$.

Before we start to prove our Theorem \ref{thm:main}, we need the following Chernoff-type bounds:

\begin{lemma}
\label{lem:chernoff1}
If $X_1,X_2,\mathellipsis,X_n$ are independent $\{0,1\}$-valued random variables with $X=\sum\nolimits_i X_i$ such that $\mathbb{E}[X]=\mu$, then we have that
\begin{enumerate}
\item {\em\cite{motwani1995randomized}} $\Pr[X<(1-\delta)\mu]\le e^{-\frac{\delta^2}{2}\mu}$\label{eq:ch1}.
\item {\em\cite{boucheron2000sharp}} $\Pr[X>\mu+\beta]\le \exp(-\frac{\beta^2}{2\mu+\frac{2}{3}\beta})$\label{eq:ch3}.
\end{enumerate}
\end{lemma}

\begin{lemma}{\em\cite{chekuri2009dependent}}
\label{lem:chernoff2}
Let $f:\{0,1\}^n\rightarrow \mathbb{R}^+$ be a monotone submodular function with marginal values in $[0,1]$. Let $\mu=\mathbb{E}[f(X_1,\mathellipsis,X_n)]$. Then for any $\delta>0$,
\[\Pr\Bigl[f(X_1,\mathellipsis,X_n)\le(1-\delta)\mu\Bigr]\le e^{-\frac{\delta^2}{2}\mu}\label{eq:ch2}.\]
\end{lemma}

We now give the following lemma:

\begin{lemma}
\label{lem:notsatisfy}
For any user $u\in \mathcal{U}$ and a non-overflowed round $k$ such that $2^k\ge t_u^*$. The probability that $H_k\cup G_k$ does not satisfy $u$ is at most 0.023.
\end{lemma}
\begin{proof}
Fix a user $u$. Consider constraints \eqref{eq:6} for $F=P_k\cap I_u$ and $t=2^k$. If $|F|\ge K_u$, $u$ is clearly satisfied in this phase because all the documents in $H_k$ are selected in this round. Therefore, we consider the case where $|F|<K_u$. From constraints \eqref{eq:6} we know that:
\[(K_u-|F|)\sum\limits_{s\in T_1(u,F)}z_{st}^*+\sum\limits_{e\in I_u\backslash F}\sum\limits_{s\in S(e,u,F)}z_{st}^*\ge y_{ut}^*(K_u-|F|)\ge \frac{1}{2}(K_u-|F|).\]

Here, either $\sum\nolimits_{s\in T_1(u,F)}z_{st}^*$ must be greater or equal to $\frac{1}{5}$, or $\sum\nolimits_{e\in I_u\backslash F}\sum\nolimits_{s\in S(e,u,F)}z_{st}^*$ must be greater or equal to $\frac{3}{10}(K_u-|F|)$.\\
\\
In the case of $\sum\nolimits_{s\in T_1(u,F)}z_{st}^*\ge \frac{1}{5}$. If there is an $s\in T_1(u,F)$ such that $50z_{st}\ge 1$, then $s$ is selected and user $u$ is satisfied.
Otherwise, since we select documents independently in our algorithm and the expected number of selected documents in $T_1(u,F)$ is
\[\mathbb{E}\Bigl[|(G_k\cup H_k)\cap T_1(u,F)|\Bigr]=\sum\limits_{s\in T_1(u,F)\backslash H_k}50z_{st}^*+|T_1(u,F)\cap H_k|\ge \sum\limits_{s\in T_1(u,F)}50z_{st}^*\ge 50\times \frac{1}{5}=10.\]

From Lemma \ref{lem:chernoff1} \eqref{eq:ch1}, we know that the probability that $H_k\cup G_k$ contains less than one document in $T_1(u,F)$ is
\[\Pr\Bigl[|(G_k\cup H_k)\cap T_1(u,F)|<(1-\frac{9}{10})\cdot 10\Bigr]\le \exp(-\frac{(\frac{9}{10})^2}{2}\cdot 10)<0.018.\]

Therefore the probability that user $u$ is not satisfied by $H_k\cup G_k$ in this case is at most 0.018.\\
\\
In the case of $\sum\nolimits_{e\in I_u\backslash F}\sum\nolimits_{s\in S(e,u,F)}z_{st}^*\ge \frac{3}{10}(K_u-|F|)$, assume boolean vector $\textbf{w}=\{w_s\}\in \{0,1\}^{|T_2(u,F)|}$ indicates the selected documents in $T_2(u,F)$. Let submodular function $f(\textbf{w})=\sum\nolimits_{e\in I_u\backslash F}\min\{1,\sum\nolimits_{s\in S(e,u,F)}w_{s}\}$, i.e. the number of topics in $I_u\backslash F$ that the selection of documents $\textbf{w}$ covers. Suppose $\overline{\textbf{z}_t}=\{\overline{z_{st}^*}\}_{s\in T_2(u,F)}$ to be a random $0/1$ vector that is obtained as follows: Independently set $\overline{z_{st}^*}$ to be 0 with probability $(1-50z_{st}^*)$ if $s\in T_2(u,F)\backslash H_k$, and 1 otherwise. Since for any $e\in I_u\backslash F$, $\sum\nolimits_{s:e\in C_s}z_{st}^*<\frac{1}{50}$ (see the definition of $F$ and $P_k$), we can find:
\[\begin{array}{lll}
\mathbb{E}\Bigl[f(\overline{\textbf{z}_t})\Bigr]&=&\sum\limits_{e\in I_u\backslash F}\Pr\Bigl[e\text{ is covered by }\overline{\textbf{z}_t}\Bigr]\\
&=&\sum\limits_{e\in (I_u\backslash F)\backslash\bigcup\nolimits_{s\in H_k}C_s}
(1-\prod\limits_{s\in S(e,u,F)}(1-50z_{st}^*))+|\{e\mid e\in (I_u\backslash F)\cap\bigcup\limits_{s\in H_k}C_s\}|\\
&\ge&\sum\limits_{e\in I_u\backslash F}(1-\prod\limits_{s\in S(e,u,F)}(1-50z_{st}^*))\\
&\ge&\sum\limits_{e\in I_u\backslash F}(1-\prod\limits_{s\in S(e,u,F)}e^{-50z_{st}^*})\\
&\ge&\sum\limits_{e\in I_u\backslash F}(1-\exp(-\sum\limits_{s\in S(e,u,F)}50z_{st}^*))\\
&\ge&\sum\limits_{e\in I_u\backslash F}(1-\frac{1}{e})\sum\limits_{s\in S(e,u,F)}50z_{st}^*\\
&\ge&(1-\frac{1}{e})15(K_u-|F|)
\end{array}\]
where the penultimate inequality is because $(1-e^{-x})\ge(1-\frac{1}{e})x$, $\forall x\in [0,1]$.

From Lemma \ref{lem:chernoff2}, we know that
\[\Pr[f(\overline{\textbf{z}_t})\le |K_u|-|F|]=\Pr[f(\overline{\textbf{z}_t})\le (1-\frac{14e-15}{15e-15})]\le e^{-(\frac{14e-15}{15e-15})^2\frac{15e-15}{2e}}<0.023.\]

This shows that the probability that user $u$ is not satisfied by $H_k\cup G_k$ in this case is at most 0.023, which complete the proof of Lemma \ref{lem:notsatisfy}.
\end{proof}
\begin{lemma}
\label{lem:overflow}
The probability that the algorithm "overflowed" in round $k$ is at most 0.03.
\end{lemma}
\begin{proof}
First we show $|H_k|\le \rho\times 2^k$. It is easy to see that $\textbf{z}=\{z_{st}\}_{s\in \mathcal{S}}$ is a feasible solution of $\LP(\SC(P_k)))$. By our assumption on \textbf{AlgoSC}, we have that
\[|H_k|\le \rho\sum\limits_{s\in \mathcal{S}}z_{st}^*=\rho\sum\limits_{s\in \mathcal{S}}\sum\limits_{t'=1}^tx_{st'}^*=\rho\sum\limits_{t'=1}^t\sum\limits_{s\in \mathcal{S}}x_{st'}^*=\rho\cdot2^k.\]
where the last equation is because $t=2^k$ in round $k$.

Therefore, it is suffice to show that $\Pr[|G_k|\ge 70\cdot 2^k]<0.03$.
In our setting,
\[\mathbb{E}\Bigl[|G_k|\Bigr]=\sum\limits_{s\in \mathcal{S}\backslash H_k}\min\{1,50z_{st}^*\}\le \sum\limits_{s\in \mathcal{S}}50z_{st}^*=\sum\limits_{s\in \mathcal{S}}\sum\limits_{t'=1}^t50x_{st'}^*=\sum\limits_{t'=1}^t\sum\limits_{s\in \mathcal{S}}50x_{st'}^*=50\cdot 2^k.\]

From Lemma \ref{lem:chernoff1} \eqref{eq:ch3}, we know that
\[\Pr\Bigl[|G_k|>50\cdot 2^k+20\cdot 2^k\Bigr]\le \exp(-\frac{400\cdot 2^{2k}}{100\cdot 2^k+\frac{40}{3}\cdot 2^k})<0.03.\]
\end{proof}
Now we will prove our Theorem \ref{thm:main}.
\begin{proof}
Let \textbf{Satisfy}($u$) denote the satisfying time of $u$ in our algorithm. From constraints \eqref{eq:1}, \eqref{eq:2} and \eqref{eq:3}, we know that $z_{sn}^*=1$ for all $s\in \mathcal{S}$, thus all the users must be satisfied after $\lceil\log n\rceil+1$ rounds. If some user $u$ is satisfied before the $\lceil\log t_u^*\rceil$th round, the satisfying time $\textbf{Satisfy}($u$)\le 2^{\lceil\log t_u^*\rceil}$.

Otherwise, if some user $u$ is satisfied after the $\lceil\log t_u^*\rceil$th round, since we select at most $(70+\rho)\cdot 2^k$ documents in each round, the satisfying time of user $u$ is at most $2\cdot (70+\rho)\cdot 2^k$ if he is satisfied in the $k$-th round. From Lemma \ref{lem:notsatisfy} and Lemma \ref{lem:overflow}, we know that the probability that user $u$ isn't satisfied after the $k$-th round where $2^k\ge t_u^*$ is less than $1-(1-0.023)\times(1-0.03)<0.053$.
Notice that the probability is independent in each round, we get the expected total satisfying time:
\[\begin{array}{lll}
\mathbb{E}\Bigl[\sum\limits_{u\in \mathcal{U}}\textbf{Satisfy}($u$)\Bigr]&\le&\sum\limits_{u\in \mathcal{U}}(2\cdot(70+\rho)\cdot 2^{\lceil\log t_u^*\rceil}+\sum\limits_{i=\lceil\log t_u^*\rceil+1}^{\lceil\log n\rceil+1}(70+\rho)\cdot 2^i\cdot 0.053^{i-\lceil\log t_u^*\rceil})\\
&\le&\sum\limits_{u\in \mathcal{U}}((140+2\rho)t_u^*+(70+\rho)t_u^*\sum\limits_{i=1}^{\infty}0.106^i)\\
&<&(149+2.12\rho)\sum\limits_{u\in \mathcal{U}}t_u^*\\
&<&(298+4.3\rho)\LPOpt =O(\rho \OPT).
\end{array}\]

This complete the proof of Theorem \ref{thm:main} from Lemma \ref{lem:OPTeqRDC}.
\end{proof}

\subsection{Solving the LP}

In order to use ellipsoid method to find the optimal solution, we need to find a polynomial-time separation oracle that verifies if a candidate solution satisfies all constraints \cite{grotschelgeometric}. Unfortunately, constraints \eqref{eq:6} contains exponentially many inequalities and it is hard to find such a separation oracle. However, we can use the trick mentioned in \cite{azar2010rankingunrelated}.
Note that in our analysis, we only consider one knapsack inequality in each round where $F=P_k\cap I_u$ and $t=2^k$. Thus if there is a solution satisfies all these $\lceil\log n\rceil+1$ inequalities in \eqref{eq:6} as well as other constraints in \eqref{eq:1}, \eqref{eq:2}, \eqref{eq:3} and \eqref{eq:5}, it is enough for our algorithm even if it is not a feasible solution for (\textbf{LP}). Therefore in each iteration of the ellipsoid method, we just need to check the validity of polynomial constraints, which forms a polynomial algorithm. 

%% file: extensions.tex
\section{Extensions}

\subsection{Ranking with Groups of Intents and Correlated Contents (\RGC)}

Now we extend \RDC\ to the problem that each user $u$ may interest in more than one sets of topics $I_{u1},I_{u2},\mathellipsis,I_{up}$, and a user is satisfied if at least one of these groups is satisfied, where $p$ is at most polynomial of $n$. Same as \RDC, a set $I_{ui}$ is satisfied if $K_{ui}$ topics in $I_{ui}$ are covered.
This time we could change our relaxed LP as follows:
\[\begin{array}{lllll}
\text{minimize}&\sum\limits_{u\in \mathcal{U}}\sum\limits_{t=1}^n(1-y_{ut})&\\
\text{subject to}&\sum\limits_{t=1}^nx_{st}=1&\forall s\in \mathcal{S}\\
&\sum\limits_{s\in \mathcal{S}}x_{st}=1&\forall t\in [n]\\
&y_{ut}\le \max_i\{g_{uit}\}&\forall u\in \mathcal{U},t\in [n]\\
&z_{st}=\sum\limits_{t'=1}^tx_{st'}&\forall t\in [n]\\
&g_{uit}(K_{ui}-|F|)\le (K_{ui}-|F|)&\sum\limits_{s\in T_1(u,i,F)}z_{st}+\sum\limits_{e\in I_u\backslash F}\sum\limits_{s:e\in C_s,s\in T_2(u,i,F)}z_{st}\\
&&\forall i\in[p],u\in \mathcal{U}, t\in [n], F\subseteq I_{ui},|F|\le K_{ui}\\
&x_{st},y_{ut},z_{st},g_{uit}\in [0,1]&\forall s\in \mathcal{S}, u\in \mathcal{U}, t\in [n], i\in[p]
\end{array}\]
where $T_1(u,i,F)$ is the set of all documents that cover at least $K_{ui}-|F|$ objects in $I_{ui}\backslash F$, $T_2(u,i,F)=\mathcal{S}\backslash T_1(u,i,F)$, and $g_{uit}$ indicates if for user $u$, group $i$ is satisfied after time $t$.

The algorithm and analysis are almost the same as in \RDC, so we won't talk about it more. Finally we get Theorem \ref{thm:RGC}.
\begin{theorem}
\label{thm:RGC}
Suppose for any $F\subseteq \mathcal{E}$, there is a polynomial time algorithm that can produce a solution for $\SC(F)$ whose cost is at most $\rho$ times the optimal value of $\LP(\SC(F))$. There is a polynomial time factor $O(\rho)$ approximation algorithm for \RGC.
\end{theorem}

\subsection{Ranking with XOS Valuations (\RXOS)}

Notice that all the problems we have mentioned in this paper are special cases of \RSI, where the users' satisfying functions are submodular functions of the set of documents.
There is another family of valuations called XOS. An XOS function is a set function which is the maximum of several additive set functions. An additive set function $f:\{0,1\}^{\mathcal{S}}\rightarrow \mathbb{R}^+\cup \{0\}$ has the form $f(F)=\sum\nolimits_{s\in F}A_s$, $\forall F\subseteq S$, where $A_s$ is a constant associated with each element $s\in S$.
Since the family of submodular functions is contained in XOS \cite{lehmann2006combinatorial}, \RSI\ is a special case of \RXOS.\footnote{The number of additive set functions which are needed to represent a submodular function may be exponential}
Suppose for each user, the XOS function contains only polynomial number of additive set functions which are non-negative. We can give an O(1)-approximation algorithm. (If the number of additive set functions is exponential, the best approximate rate we can hope for is $O(\log n)$ since \RXOS\ generates \RSI.)

For each user $u$, suppose the additive set functions be $f_{ui},i=1,2,\ldots,p$, where $f_{ui}(F)=\sum\nolimits_{s\in F}A_{uis}$ for $F\subseteq \mathcal{S}$. Without lose of generality, we can let the satisfying time of user $u$ be $t_u=\min\{t\mid\max_i{f_{ui}(\text{the first }t\text{ selected documents})}\ge1\}$. Now we could have the following LP relaxation:
\[\begin{array}{lllll}
\text{minimize}&\sum\limits_{u\in \mathcal{U}}\sum\limits_{t=1}^n(1-y_{ut})&\\
\text{subject to}&\sum\limits_{t=1}^nx_{st}=1&\forall s\in \mathcal{S}\\
&\sum\limits_{s\in \mathcal{S}}x_{st}=1&\forall t\in [n]\\
&y_{ut}\le \max_i\{g_{uit}\}&\forall u\in \mathcal{U},t\in [n]\\
&z_{st}=\sum\limits_{t'=1}^tx_{st'}&\forall t\in [n]\\
&g_{uit}(1-\sum\limits_{s\in F}A_{uis})\le (1-\sum\limits_{s\in F}A_{uis})&\sum\limits_{s\in T_1(u,i,F)}z_{st}+\sum\limits_{s\in T_2(u,i,F)}A_{uis}z_{st}\\
&&\forall i\in[p], u\in \mathcal{U}, t\in [n], F\subseteq \mathcal{S},\sum\limits_{s\in F}A_{uis}\le 1\\
&x_{st},y_{ut},z_{st},g_{uit}\in [0,1]&\forall s\in \mathcal{S}, u\in \mathcal{U}, t\in [n], i\in[p]
\end{array}\]
where $T_1(u,i,F)=\{s\mid A_{uis}\ge(1-\sum\nolimits_{e\in F}A_{uie}),s\in \mathcal{S}\}$, $T_2(u,i,F)=\mathcal{S}\backslash T_1(u,i,F)$, and $g_{uit}$ indicates if $f_{ui}(\{\text{the first }t\text{ selected documents}\})\ge1$.

This time we do not need to consider the set cover instances, and the $k$-th round of our algorithm can be:
\begin{itemize}
\item Let $t=2^k$, $G_k=\emptyset$.
\item For each $s\in \mathcal{S}$, add document $s$ to $G_k$ independently with probability $\min\{1,50z_{st}^*\}$
\item If there are more than $70\cdot 2^k$ documents in $G_k$, we say this round is "overflowed" and select nothing, else we select all the documents in $G_k$ in arbitrary order in this round.
\end{itemize}

All the discussions are the same as in section \ref{subsec:analysis} except there is no $H_k$, and in the case that $\sum\nolimits_{s\in T_2(u,i,F)}A_{uis}z_{st}\ge\frac{3}{10}(1-\sum\nolimits_{s\in F}A_{uis})$, we could have $\mathbb{E}\Bigl[\sum\nolimits_{s\in T_2(u,i,F)}A_{uis}\overline{z_{st}^*}\Bigr]\ge 15(1-\sum\nolimits_{s\in F}A_{uis})$ and use Lemma \ref{lem:chernoff2} directly.

Finally we can have the following theorem:
\begin{theorem}
Suppose for each user, the XOS function contains only polynomial additive set functions which are non-negative. There is an $O(1)$-approximation for RXOS.
\end{theorem}